\documentclass{CSML}
\pdfoutput=1
\usepackage{hyperref}
\hypersetup{hidelinks}

\usepackage{lastpage}
\newcommand{\dOi}{13(3:12)2017}
\lmcsheading{}{1--\pageref{LastPage}}{}{}%
{Jan.~10, 2017}{Aug.~15, 2017}{}

\usepackage{amsmath,amssymb}
\usepackage[all]{xy}
\usepackage{latexsym}
\usepackage{amsthm,color}
\usepackage{amsmath,amscd,verbatim}
\usepackage{hyperref}
\hypersetup{hidelinks}
\usepackage{graphicx}

\input diagxy

\theoremstyle{plain}
\theoremstyle{definition}
  \newtheorem{defn}[thm]{Definition}
  
  \newtheorem{exmp}[thm]{Example}

\makeatletter
\def\ps@pprintTitle{%
 \let\@oddhead\@empty
 \let\@evenhead\@empty
 \def\@oddfoot{\centerline{\thepage}}%
 \let\@evenfoot\@oddfoot}
\makeatother

\begin{document}

\newcommand{\oto}{{\to\hspace*{-3.1ex}{\circ}\hspace*{1.9ex}}}
\newcommand{\lam}{\lambda}
\newcommand{\da}{\downarrow}
\newcommand{\Da}{\Downarrow\!}
\newcommand{\D}{\Delta}
\newcommand{\ua}{\uparrow}
\newcommand{\ra}{\rightarrow}
\newcommand{\la}{\leftarrow}
\newcommand{\lra}{\longrightarrow}
\newcommand{\lla}{\longleftarrow}
\newcommand{\rat}{\!\rightarrowtail\!}
\newcommand{\up}{\upsilon}
\newcommand{\Up}{\Upsilon}
\newcommand{\ep}{\epsilon}
\newcommand{\ga}{\gamma}
\newcommand{\Ga}{\Gamma}
\newcommand{\Lam}{\Lambda}
\newcommand{\CF}{{\cal F}}
\newcommand{\CA}{{\mathcal{A}}}
\newcommand{\CG}{{\cal G}}
\newcommand{\CH}{{\cal H}}
\newcommand{\CN}{{\mathcal{N}}}
\newcommand{\CB}{{\cal B}}
\newcommand{\CT}{{\cal T}}
\newcommand{\CS}{{\cal S}}
\newcommand{\CP}{{\cal P}}
\newcommand{\CU}{\mathcal{U}}
\newcommand{\CW}{\mathcal{W}}
\newcommand{\CQ}{\mathcal{Q}}
\newcommand{\mq}{\mathcal{Q}}
\newcommand{\cu}{{\underline{\cup}}}
\newcommand{\ca}{{\underline{\cap}}}
\newcommand{\nb}{{\rm int}}
\newcommand{\Si}{\Sigma}
\newcommand{\si}{\sigma}
\newcommand{\Om}{\Omega}
\newcommand{\bm}{\bibitem}
\newcommand{\bv}{\bigvee}
\newcommand{\bw}{\bigwedge}
\newcommand{\dda}{\downdownarrows}
\newcommand{\dia}{\diamondsuit}
\newcommand{\y}{{\bf y}}
\newcommand{\colim}{{\rm colim}}
\newcommand{\fR}{R^{\!\forall}}
\newcommand{\eR}{R_{\!\exists}}
\newcommand{\dR}{R^{\!\da}}
\newcommand{\uR}{R_{\!\ua}}
\newcommand{\swa}{{\swarrow}}
\newcommand{\sea}{{\searrow}}
\newcommand{\bbA}{{\mathbb{A}}}
\newcommand{\frX}{{\mathfrak{X}}}
\newcommand{\frx}{{\mathfrak{x}}}
\newcommand{\frY}{{\mathfrak{Y}}}
\newcommand{\fry}{{\mathfrak{y}}}
\newcommand{\frZ}{{\mathfrak{Z}}}
\newcommand{\frz}{{\mathfrak{z}}}

\newcommand{\id}{{\rm id}}
\newcommand{\bbU}{{\mathbb{U}}}
\newcommand{\bbP}{{\mathbb{P}}}
\newcommand{\bbT}{{\mathbb{T}}}
\newcommand{\bbS}{{\mathbb{S}}}
\newcommand{\CV}{{\mathcal{V}}}
\newcommand{\sU}{{\sf{U}}}
\newcommand{\sV}{{\sf{V}}}
\newcommand{\sW}{{\sf{W}}}
\newcommand{\sP}{{\sf P}}
\newcommand{\sy}{{\sf{y}}}
\newcommand{\sk}{{\sf{k}}}
\newcommand{\sfs}{{\sf{s}}}
\newcommand{\h}{\text{-}}

\numberwithin{equation}{section}
\renewcommand{\theequation}{\thesection.\arabic{equation}}

%\begin{frontmatter}
\title[A Note on the Topologicity of Quantale-Valued Topological Spaces]{A Note on the Topologicity of Quantale-Valued Topological Spaces{\rsuper*}
}

\author{Hongliang Lai}
\address{School of Mathematics, Sichuan University, Chengdu 610064, China}
\email{hllai@scu.edu.cn}

\author{Walter Tholen} % corresponding author
\address{Department of Mathematics and Statistics, York University, Toronto, Canada}
\email{tholen@mathstat.yorku.ca}

\dedicatory{Dedicated to Ji\v{r}\'i Ad\'amek, admired mathematician and friend}

%\cortext[cor]{Corresponding author.}

\titlecomment{{\lsuper*}Partial financial assistance by National Natural Science Foundation of China (11101297), International Visiting Program for Excellent Young Scholars of Sichuan University, and by the Natural Sciences and Engineering Research Council (NSERC) of Canada is gratefully acknowledged. This work was initiated while the first author held a Visiting Professorship at York University in 2015-16. It was presented in part by the second author at a session of the {\em Peripatetic Seminar on Sheaves and Logic}, held in March 2016 at the Technical University of Braunschweig to mark 
Ji\v{r}\'i Ad\'amek's retirement, and it was finalized while the second author was visiting the Centre of Australian Category Theory at Macquarie University in December 2016, with the kind support and hospitality of its members.}

\begin{abstract}
For a quantale $\sV$, the category $\sf V$-{\bf Top} of $\sV$-valued topological spaces may be introduced as a full subcategory of those $\sV$-valued closure spaces whose closure operation preserves finite joins. In generalization of Barr's characterization of topological spaces as the lax algebras of a lax extension of the ultrafilter monad from maps to relations of sets, for $\sV$ completely distributive, $\sV$-topological spaces have recently been shown to be characterizable by a lax extension of the ultrafilter monad to $\sV$-valued relations. As a consequence, $\sV$-{\bf Top} is seen to be a topological category over {\bf Set}, provided that $\sV$ is completely distributive. In this paper we give a choice-free proof that $\sV$-{\bf Top} is a topological category over $\bf Set$ under the considerably milder provision that $\sV$ be a spatial coframe. When $\sV$ is a continuous lattice, that provision yields complete distributivity of $\sV$ in the constructive sense, hence also in the ordinary sense whenever the Axiom of Choice is granted.
\end{abstract}

\maketitle
%\begin{keyword}
%%% keywords here, in the form: keyword \sep keyword
%quantale\sep $\sV$-valued closure space \sep $\sV$-valued topological space, topological category, spatial coframe, contiunous lattice, completely distributive lattice.
%%\sep discrete $\sV$-presheaf monad \sep lax distributive law \sep lax $(\lam,\sV)$-algebra \sep probabilistic approach space, algebraic functor, change-of-base functor.
%\MSC[2010] 54A05, 54B30, 18B30,  18D99, 06D10, 54E70.
%\end{keyword}

%\end{frontmatter}

\section{Introduction}
Trivially, the category {\bf Top} of topological spaces may be considered as a full subcategory of the category {\bf Cls} of closure spaces, given by those closure spaces $X$ for which the closure operation $c:\sP X\to\sP X$ preserves finite unions. Non-trivially, in \cite{Barr}, Barr showed that {\bf Top} is isomorphic to the category of lax Eilenberg-Moore algebras of the ultrafilter monad $\mathbb U$, laxly extended from {\bf Set} to the category {\bf Rel} of sets and relations. In the language of monoidal topology \cite{MonTop}, this latter category is the category $(\mathbb U,{\sf 2})$-{\bf Cat} of small $(\mathbb U,{\sf 2})$-enriched categories, with $\sf 2$ denoting the two-element chain, considered as a quantale, while {\bf Cls} is the category $(\mathbb P,{\sf 2})$-{\bf Cat}, with $\mathbb P$ denoting the power set monad, suitably extended from {\bf Set} to {\bf Rel}.

In \cite{LT} the authors replaced the quantale $\sf 2$ by an arbitrary quantale $\sV$ and revisited the category $\sV$-{\bf Cls} of $\sV$-valued closure spaces, comprehensively studied earlier in \cite{Seal2009} when $\sV$ is completely distributive, and then considered its full subcategory $\sV$-{\bf Top} of $\sV$-valued topological spaces. (We caution the reader that these terms appear in the literature also for rather different concepts; see Remark \ref{caution}(2) below.) For $\sV$ the Lawvere quantale $[0,\infty]$ (see \cite{Lawvere}), $\sV$-valued topological spaces are exactly approach spaces as introduced by Lowen \cite{Lowen} in terms of a point-set distance. For $\sV$ the quantale $\bf{\Delta}$ of distance distribution functions, they are probabilistic approach spaces, as considered recently in \cite{JagerApp}, which generalize Menger's \cite{Menger1942} ``statistical metric spaces", just as approach spaces generalize metric spaces; see also \cite{BrockKent, HofmannReis}.
The main result of \cite{LT} confirms that, when the quantale $\sV$ is completely distributive, the Barr representation of topological spaces remains valid at the $\sV$-level once the power set monad and the ultrafilter monad are suitably extended from {\bf Set} to the category $\sV$-{\bf Rel} of sets and $\sV$-valued relations. Briefly: the category $\sV$-{\bf Top}, defined as the full subcategory of $(\mathbb P,\sV)$-{\bf Cat} given by those $\sV$-valued closure spaces whose structure preserves finite joins, is isomorphic to $(\mathbb U,\sV)$-{\bf Cat} -- provided that $\sV$ is completely distributive.

Since $(\mathbb T,\sV)$-{\bf Cat} is easily seen to be a topological category (in the sense of  \cite{AHS}) over {\bf Set}, for any laxly extended {\bf Set}-monad $\mathbb T$ (see \cite{MonTop}), as a byproduct of the equivalence result of \cite{LT} one obtains that $\sV$-{\bf Top} is topological over {\bf Set} whenever $\sV$ is completely distributive. The question was posed by Dexue Zhang and Lili Shen whether topologicity may be confirmed without the provision of complete distributivity. In this paper we give a partial answer to this question, and we do so without invoking the Axiom of Choice, by showing that $\sV$-{\bf Top} lies bicoreflectively in the topological category $\sV$-{\bf Cls} and, hence, is topological itself -- provided that every element in {\sV} is the join of a set of coprime elements. This provision is equivalent to the complete lattice $\sV$ being a spatial coframe,
{\em i.e.}, being isomorphic to the lattice of closed sets of some topological space (see \cite{Johnstone1982}); in particular, binary joins must distribute over arbitrary infima in $\sV$. In the presence of the Axiom of Choice one can show that complete distributivity of a complete lattice $\sV$ is equivalent to $\sV$ being continuous and a spatial coframe (see \cite{Gierz2003}).

In the next section we recall the definitions and main examples of $\sV$-valued closure spaces and $\sV$-valued topological spaces, with a novel take on the transitivity/idempotency axiom for a $\sV$-valued closure operation which turns out to be useful in what follows. Section 3 recalls some known facts on spatial coframes vis-\'a-vis complete distributivity, but we do so being careful to avoid the Axiom of Choice to the extent possible. The main result of the paper is given in Section 4, where we establish the coreflector of $\sV$-{\bf Cls} onto $\sV$-{\bf Top}, by mimicking the construction of the additive core of a categorical closure operator, as given in \cite{DT}. Finally, in Section 5, we show that
the structure of a $\sV$-valued closure space and, hence, also the structure of a $\sV$-valued topological space, may be equationally defined within the category $\sV$-{\bf Cat} of $\sV$-categories, which is somewhat surprising when one considers the fact that the axioms governing reflexivity/extensitivity and transitivity/idempotency are given in terms of inequalities, rooted in the order of $\sV$.

\section{$\sV$-valued topological spaces}
Throughout the paper, let $\sV=(\sV,\otimes,{\sf k})$ be a (unital, but not necessarily commutative) {\em quantale, i.e.,} a complete lattice with a monoid structure whose binary operation $\otimes$ preserves suprema in each variable. We make no additional provisions for the  tensor-neutral element $\sk$
vis-\`{a}-vis the bottom and top elements in $\sV$; in particular, we  exclude neither the case $\sk=\bot$ (so that $|\sV|=1$, {\em i.e.}, $\sV$ may be {\em trivial}), nor $\sk < \top$ ({\em i.e.}, $\sV$ may fail to be {\em integral}). ${\sf P}X$ denotes the power set of the set $X$, and $\sV^X$ is the set of maps $X\to\sV$.

We use the following simplification of the key definition of \cite{LT} which, in turn, builds on the equivalent treatment of $\sV$-valued closure spaces given in \cite{Seal2009}:
\begin{defn} A {\em $\sV$-valued closure space} is a set $X$ equipped with a map $c: \sP X\to \sV^X$ satisfying
the reflexivity and transitivity conditions

\begin{enumerate}
\item[{\rm (R)}]$\quad \forall x\in A\subseteq X:$ \;\;$ \sk\leq (cA)(x)$,
\item[{\rm (T)}]$\quad\forall A, B\subseteq X, x\in X:$ \;\,$ \big(\bw_{y\in B}(cA)(y)\big)\otimes (cB)(x)\leq (cA)(x)$.
\end{enumerate}

$(X,c)$ is a {\em $\sV$-valued topological space} if, in addition, $c:\sP X\to \sV^X$ is {\em finitely additive}, {\em i.e.}, preserves finite joins:

\begin{enumerate}
\item[{\rm (A)}]$\quad \forall A,B\subseteq X,\,x\in X: \quad (c\emptyset)(x)=\bot\quad\text{ and }\quad c(A\cup B)(x)=(cA)(x)\vee (cB)(x).$
\end{enumerate}

 A map $f:X\to Y$ of $\sV$-closure spaces $(X,c), (Y,d)$ is {\em continuous} (or, depending on context, {\em contractive}) if

 \begin{enumerate}
 \item[{\rm (C)}] $\quad \forall A\subseteq X,\, x\in X:\quad (cA)(x)\leq d(fA)(fx).$
 \end{enumerate}

We obtain
the category $\sV\text{-}{\bf Cls}$  of $\sV$-valued closure spaces and its full subcategory $\sV\text{-}{\bf Top}$ of $\sV$-valued topological spaces, and their continuous maps.
\end{defn}

\begin{rem}\label{caution} \leavevmode
  \begin{enumerate}
\item A $\sV$-valued closure space structure $c$ on $X$ satisfies the monotonicity condition $(\emptyset\neq B\subseteq A\subseteq X\Longrightarrow cB\leq cA)$. If $\sV$ is integral (so that ${\sf k}=\top$), or if $c$ is finitely additive, then the restriction $B\neq\emptyset$ is, of course, not needed.

\item The notion of $\sV$-valued closure operator as used here must be carefully distinguished  from the $\sV$-closure operators considered in other papers, notably in \cite{Demirci2007}. Rather than investigating operators $c:\sP X\to\sV^X$, which simply generalize, from $\sf{2}$ to $\sV$, the truth value of membership in the closure of a (standard) subset of $X$, \cite{Demirci2007} studies operators $c:\sV^X\to\sV^X$, thus relativizing also the arguments of $c$. Similarly, in \cite{Hohle2001} and other literature, the term $\sV$-valued topological space has a meaning very different from the one used here. Indeed, H\"{o}hle \cite{Hohle2001} relativizes the open-subset concept through the study of certain substructures $\tau\subseteq\sV^X$; as further examples of variations of this type of approach, we refer also to \cite{HohleSostak1999, Rodabough2007, Demirci2014}. 
\end{enumerate}
\end{rem}

\begin{exmp}
  \leavevmode
  \begin{enumerate}
\item For the terminal quantale {\sf 1} one has ${\sf 1}\text{-}{\bf Cls}={\sf 1}\text{-}{\bf Top}\cong{\bf Set}$.

\item For the two-element chain ${\sf 2}=\{\bot<\top\}$, considered as a quantale $({\sf 2},\wedge,\top)$, under the identification ${\sf 2}^X=\sP X$ conditions {\rm (R)} and {\rm  (T)} read respectively as $A\subseteq cA$
and $(B\subseteq cA\Rightarrow cB\subseteq cA)$, for all $A,B\subseteq X$, where, in the presence of {\rm (R)}, condition {\rm (T)}
breaks down to the conjunction of the monotonicity and idempotency conditions $(B\subseteq A\Rightarrow cB\subseteq cA)$ and $ccA\subseteq cA$, for all $A,B\subseteq X$. With {\rm (A)} and {\rm (C)} translating to $c\emptyset=\emptyset, \, c(A\cup B)=cA\cup cB$ and
$f(cA)\subseteq c(fA)$ for all $A,B\subseteq X$, one obtains respectively ${\sf 2}\text{-}{\bf Cls}={\bf Cls}$ and ${\sf 2}\text{-}{\bf Top}\cong{\bf Top}$, {\em i.e.}, the standard categories of closure spaces and of topological spaces, as described by closure operations. The presentation of closure spaces as, in the language of \cite{MonTop}, $(\mathbb P,\sf 2)$-categories (with $\mathbb P$ the powerset monad), goes back to \cite{Seal2005}; instead of $\mathbb P$ one may also use the ``up-set monad"  \cite{Seal2009}.

\item For any (multiplicatively written) monoid $M$ with neutral element $\eta$, we consider the quantale freely generated by $M$; it is given by the power set $\sP M$, ordered by inclusion and provided with the tensor product that extends the multiplication of $M$ to its subsets: $AB=\{\alpha\beta\,|\,\alpha \in A,\, \beta \in B\}$. Note that, since $\{\eta\}$ is neutral in $\sP M$, this quantale is integral (so that $\{\eta\}$ is its top element) only if $M$ is trivial, and it is commutative only if $M$ is. Since maps $\sP X\to(\sP M)^X\cong (\sP X)^M$ correspond to maps $\sP X\times M\to \sP X$, defining
$$ x\in A\cdot\alpha:\Longleftrightarrow\alpha\in (cA)(x)$$
for all $x\in X, \,A\subseteq X,\, \alpha\in M$, we can rewrite a $\sV$-valued closure space structure $c$ on $X$ as a lax right action $\cdot$ of the monoid $M$ on the ordered set $\sP X$. Indeed, conditions {\rm (R)} and {\rm (T)} read as
$$A\subseteq A\cdot\eta\quad\text{and}\quad (B\subseteq A\cdot\alpha\Longrightarrow B\cdot\beta\subseteq A\cdot(\alpha\beta))$$
or, equivalently, as
$$A\subseteq A\cdot\eta,\quad (A\cdot\alpha)\cdot\beta\subseteq A\cdot(\alpha\beta)\quad\text{and}\quad (B\subseteq A\Longrightarrow B\cdot\beta\subseteq A\cdot\beta),$$
for all $A,B\subseteq X,\, \alpha,\beta\in M$.
Continuity of a map $f:X\to Y$ of $\sP M$-valued closure spaces amounts to lax preservation of the lax right action:
$f(A\cdot \alpha)\subseteq (fA)\cdot\alpha$, for all $A\subseteq X, \,\alpha\in M$. For a $\sP M$-valued topological space, all translations $(-)\cdot\alpha:\sP X\to\sP X$ must preserve finite unions.

Note that, since $A\mapsto A\cdot\eta$ is a closure operation on X, one has a functor $\sP M\text{-}{\bf Cls}\to{\bf Cls}$ that restricts to $\sP M\text{-}{\bf Top}\to {\bf Top}$.

\item For the Lawvere quantale $(([0,\infty],\geq),+,0)$, using the point-set-distance function $\delta: X\times \sP X\to [0,\infty]$ with $\delta(x,A)=(cA)(x)$, we may re-state the above conditions as
\begin{enumerate}
\item [{\rm (R)}] $\quad\forall x\in A\subseteq X:\quad\delta(x,A)=0$,
\item [{\rm (T)}] $\quad\forall x\in X,\,A,B\subseteq X:\quad\delta(x,A)\leq {\rm sup}_{y\in B}\delta(y,A)+\delta(x,B)$,
\item [{\rm (A)}] $\quad\forall x\in X,\, A,B\subseteq X:\quad\delta(x,\emptyset)=\infty\quad$ and $\quad\delta(x,A\cup B)={\rm min}\{\delta(x,A),\delta(x,B)\},$
\item [{\rm (C)}] $\quad\forall x\in X,\,A\subseteq X:\quad \delta(fx,fA)\leq \delta(x,A).$
\end{enumerate}
The resulting category $[0,\infty]\text{-}{\bf Top}$ is the category $\bf App$ of {\em approach spaces}, as introduced by Lowen \cite{Lowen, Lowen2015} under a slight, but equivalent, variation of condition (T). The ambient category $[0,\infty]\text{-}{\bf Cls}$ was considered in \cite{Seal2005}.

\item Let $\&$ be a commutative monoid operation on $[0,1]$ with its natural order, preserving suprema in each variable (also known as a {\em left-continuous t-norm} on [0,1]), and having 1 as its neutral element -- such as the ordinary multiplication $\times$ of real numbers, the \L ukasiewicz operation $\alpha\&\beta=\rm{max}\{\alpha+\beta-1,0\}$, or the frame operation $\alpha\&\beta=\rm{min}\{\alpha,\beta\}$. For the quantale $[0,1]_{\&}=(([0,1],\leq),\&,1)$ we may then consider its coproduct ${\bf{\Delta}}_{\&}$ with the Lawvere quantale $[0,\infty]$ in the category of commutative quantales and their homomorphisms (= sup-preserving homomorphisms of monoids), which may be described as follows (see \cite{LaxDistLaws, GHK, LT}). (Note that, of course, $[0,1]_{\times}$ is isomorphic to $[0,\infty]$.) The underlying
set $\bf{\Delta}$ of ${\bf{\Delta}}_{\&}=({\bf{\Delta}},\odot,\kappa)$
of all {\em distance distribution functions} $\varphi: [0,\infty]\to[0,1]$, required to satisfy the left-continuity condition $\varphi(\beta)={\rm sup}_{\alpha<\beta}\varphi(\alpha)$ for all $\beta\in [0,\infty]$, inherits its order from $[0,1]$, and its monoid structure is given by the commutative
convolution product
\[(\varphi\odot\psi)(\gamma)={\rm sup}_{\alpha+\beta\leq\gamma}\varphi(\alpha)\&\psi(\beta).\]
The $\odot$-neutral function $\kappa$ satisfies $\kappa(0)=0$ and $\kappa(\alpha)=1$ for all $\alpha >0$. We note that $\kappa=\top$ in ${\bf{\Delta}}_{\&}$ (so ${\bf{\Delta}}_{\&}$ is integral), while the bottom element in ${\bf {\Delta}}_{\&}$ has constant value $0$; we write $\bot=0$.
With the quantale homomorphisms
$\sigma:[0,\infty]\to {\bf {\Delta}}_{\&}$
and $\tau:[0,1]_{\&}\to {\bf {\Delta}}_{\&}$,
defined by $\sigma(\alpha)(\gamma)=0\text{ if }\gamma\leq\alpha\text{, and }1$ otherwise, and
$\tau(u)(\gamma)=u\text{ if }\gamma>0\text{, and }0$ otherwise, every $\varphi\in {\bf{\Delta}}$ has a presentation
$$\varphi=\bv_{\alpha\in[0,\infty]}\sigma(\alpha)\odot\tau(\varphi(\alpha))
=\bv_{\alpha\in(0,\infty)}\sigma(\alpha)\odot\tau(\varphi(\alpha)),$$
which then shows that $\sigma,\tau$ serve as the coproduct injections of ${\bf{\Delta}}_{\&}$. In terms of a point-set-distance-distribution function $\delta:X\times\sP X\to {\bf \Delta}$, the relevant conditions describing the category
${\bf{\Delta}}_{\&}\text{-}{\bf Top}\cong {\bf ProbApp}_{\&}$ of {\em $\&$-probabilistic approach spaces} \cite{Jager, JagerApp} read as

\begin{enumerate}
\item [{\rm (R)}] $\quad\forall x\in A\subseteq X:\quad\delta(x,A)=\kappa$,
\item [{\rm (T)}] $\quad\forall x\in X,\,A,B\subseteq X:{\rm inf}_{y\in B}\delta(y,A)\odot\delta(x,B)\leq\delta(x,A)$,
\item [{\rm (A)}] $\quad\forall x\in X,\, A,B\subseteq X:\quad\delta(x,\emptyset)=0\quad$ and $\quad\delta(x,A\cup B)={\rm max}\{\delta(x,A),\delta(x,B)\},$
\item [{\rm (C)}] $\quad\forall x\in X,\,A\subseteq X:\quad \delta(x,A)\leq\delta(fx,fA).$
\end{enumerate}
\end{enumerate}
\end{exmp}

In the next section it will be convenient to use the following notation for any map $c: \sP X\to \sV^X$, which suggests itself by Seal's description \cite{Seal2009} of $\sV$-valued closure spaces: for all $A\subseteq X, x\in X$ and $v\in\sV$ we put
$$(\overline{c}A)(x):=\bv_{v\in\sV}v\otimes c(c^vA)(x),$$
with $c^vA:=\{z\in X\,|\,v\leq (cA)(z)\}.$ The map $\overline{c}$ plays the role of the ``composite of $c$ with itself"; indeed, we can reformulate (T), as follows:

\begin{lem}\label{cc}
Let $c: \sP X\to\sV^X$ be monotone. Then $c\leq \overline{c}$ if $c$ satisfies $\rm (R)$, and $c$ satisfies $\rm (T)$ if, and only if, $\overline{c}\leq c$. Furthermore, for any map $d:\sP X\to\sV^X$ with $d\leq c$ one has $\overline{d}\leq\overline{c}$.
\end{lem}

\begin{proof}
With the monotonicity of $c$ one obtains from (R)
$$ (cA)(x)\leq {\sf k}\otimes c(c^{\sf k}A)(x)\leq\bv_{v\in\sf V}v\otimes c(c^vA)(x)=(\overline{c}A)(x),$$
for all $A\subseteq X, x\in X$. If $c$ satisfies (T), for every $v\in\sV$ one considers $B:=c^vA$, so that $v\leq(cA)(y)$ for all $y\in B$. Then
$$ v\otimes (cB)(x)\leq (\bw_{y\in B}(cA)(y))\otimes(cB)(x)\leq (cA)(x),$$
and $(\overline{c}A)(x)\leq (cA)(x)$ follows. Conversely, assuming $\overline{c}\leq c$, for $A, B, x$ as in (T) one considers $v:=\bw_{y\in B}(cA)(y)$. Then $B\subseteq c^vA$, and with the monotonicity of $c$ one concludes
$$v\otimes (cB)(x)\leq v\otimes c(c^vA)(x)\leq (\overline{c}A)(x)\leq (cA)(x),$$
that is: (T). Finally, for $d\leq c$ one trivially has $d^vA\subseteq c^vA$ for all $v\in\sV$ and, hence,
\[
(\overline{d}A)(x)=\bv_{v\in\sV}v\otimes d(d^vA)(x)\leq\bv_{v\in \sV}v\otimes
c(d^vA)(x)\leq\bv_{v\in\sV}v\otimes c(c^vA)(x)=(\overline{c}A)(x).
\tag*{\qEd}
\]
\def\popQED{}
\end{proof}

Recall that an element $p$ in a poset $L$ is {\em coprime} when $p\leq\bv F$ with $F\subseteq L$ finite always gives some $x\in F$ with $p\leq x$; that is: when $p> \bot$ and, for all $u,v\in L$, one has $p\leq u\vee v$ only if $p\leq u$ or $p\leq v$. The poset $L$ is said to be {\em sup-generated by its coprime elements} if every element is the supremum of a set of coprime elements in $L$. We will shed light on the status of this property in the next section. Here we just use it to prove the following lemma and proposition.

\begin{lem}\label{coprime gen}
If $\sV$ is sup-generated by its coprime elements and $c$ is monotone, then
\[(\overline{c}A)(x)=\bv_{p\in\sV\; {\rm coprime}}p\otimes c(c^pA)(x)\]
for all $A\subseteq X, x\in X$.
\end{lem}

\begin{proof}
By hypothesis, every $v\in\sV$ can be written as $v=\bv_{p\leq v\;{\rm coprime}}p$, and with the monotonicity one obtains
\[v\otimes c(c^vA)(x)=\bv_{p\leq v\;{\rm coprime}}p\otimes c(c^vA)(x)\leq\bv_{p\in\sV\;{\rm coprime}}p\otimes c(c^pA)(x),\]
which shows $``\leq"$ of the desired equality; $``\geq"$ holds trivially.
\end{proof}

Let us use the following auxiliary notion and call $(X,c)$ a $\sV$-{\em valued pretopological space} if the map $c:\sP X\to\sV^X$ satisfies (R) and (A). With Lemma \ref{coprime gen} one easily sees that these properties survive the passage from $c$ to $\overline{c}$, as follows.

\begin{prop}\label{addpres}
For a $\sV$-valued pretopological space $(X,c)$, when $\sV$ is sup-generated by its coprime elements,  $(X,\overline{c})$ is also a $\sV$-valued pretopological space.
\end{prop}

\begin{proof}
(R) follows from Lemma \ref{cc}, and for (A) we first note

$$(\overline{c}\emptyset)(x)=\bv_{v\in\sV}v\otimes c(c^v\emptyset)(x)=\bv_{v\in\sV}v\otimes\bot=\bot$$
for all $x\in X$. Furthermore, since for $p\in\sV$ coprime one obviously has $c^p(A\cup B)=(c^pA)\cup(c^pB)$ whenever $A,B\subseteq X$, we obtain with Lemma \ref{coprime gen}
\begin{align*}
\overline{c}(A\cup B)(x) &= \bv_{p}p\otimes c(c^p(A\cup B))(x) \\
&= \bv_pp\otimes (c(c^pA)(x)\vee c(c^pB)(x))\\
&= (\bv_pp\otimes c(c^pA)(x))\vee(\bv_pp\otimes c(c^pB)(x))\\
&= (\overline{c}A)(x)\vee(\overline{c}B)(x).
\tag*{\qEd}
\end{align*}
\def\popQED{}
\end{proof}

\begin{rem}\label{VTopasPVCat}
  \leavevmode
\begin{enumerate}
\item For $\sV={\sf 2}$, $\sV$-valued pretopological spaces are precisely the usual {\em pretopological spaces} (also known as {\em \v{C}ech-topological spaces} \cite{Cech1966}; see also \cite{DT}), and for $\sV=[0,\infty]$ they go by the name {\em pre-approach spaces} (see \cite{ColebundersVerbeeck2000}).

\item In \cite{LT}, drawing on \cite{Seal2009}, various alternative, but equivalent descriptions of $\sV$-valued closure spaces and topological spaces are provided. First of all, a $\sV$-valued closure space structure $c$ on $X$ gives a family of maps
$(c^v:\sP X \to \sP X)_{v\in\sV}$ satisfying
\begin{enumerate}[leftmargin=10mm]
\item[{\rm (C0)}]
$\text{if  }B\subseteq A, \text{ then  }c^vB\subseteq c^vA,$
\item[{\rm (C1)}] $\text{if  }v\leq\bv_{i\in I}u_i, \text{ then   }\bigcap_{i\in I}c^{u_i}A\subseteq c^vA,
$\item[{\rm (C2)}] $A\subseteq c^{\sk}A,$
\item [{\rm (C3)}] $c^uc^vA\subseteq c^{v\otimes u}A,$
\end{enumerate}
for all $A, B\subseteq X$ and $ u,v,u_i \in V\; (i\in I)$.
Conversely, for any family of maps $c^v:\sP X\to\sP X\;(v\in \sV)$ satisfying the conditions {\rm (C0)--(C3)}, putting $$(cA)(x):=\bv\{v\in\sV\;|\;x\in c^vA\}\quad (A\subseteq X,\;x\in X)$$ makes $(X,c)$ a $\sV$-valued closure space,
and the two processes are inverse to each other. Under this bijection, when $\sV$ is completely distributive so that, in particular,  $\sV$ is generated by its coprime elements (see \cite{Gierz2003}),  $\sV$-valued topological structures are characterized by
\[c^p\emptyset=\emptyset \;{\rm and}\; c^p(A\cup B)=c^pA\cup c^pB\]
for all coprime elements $p$ in $\sV$ and $A,B\subseteq X$.

\item Next, $\sV$-valued closure spaces are equivalently presented as $({\mathbb P},\sV)$-{\em categories} in the sense of \cite{MonTop}, with the powerset monad $\mathbb P$ on $\bf Set$ laxly extended to the category $\sV$-{\bf Rel} of sets and $\sV$-valued relations by
\[\hat{\sP}r(A,B)=\bw_{y\in B}\bv_{x\in A}r(x,y),\]
for all $\sV$-relations $r:X\nrightarrow Y, \, A,B\subseteq X$. Furthermore, when $\sV$ is completely distributive, $\sV$-valued topological spaces are equivalently presented as $(\mathbb U,\sV)$-categories, with $\mathbb U$ denoting the ultrafilter monad on $\bf Set$, laxly extended to  $\sV$-{\bf Rel} by
\[\overline{\sU}r(\frx,\fry)=\bw_{A\in\frx,B\in\fry}\bv_{x\in A,y\in B}r(x,y),\]
for all $r:X\nrightarrow Y,\;\frx\in\sU X,\fry\in\sU Y:$ see \cite{LT} for details. Of course, these bijective correspondences pertain also to the relevant morphisms and therefore give isomorphisms of categories that commute with the underlying $\bf Set$ functors.
\end{enumerate}
\end{rem}

\section{Some known properties of spatial coframes and continuous lattices}

Since in the following section we will heavily rely on the property encountered in Lemma \ref{coprime gen} and Proposition \ref{addpres}, in this section we recall some well-known facts on lattices that are sup-generated by their coprime elements.

\begin{rem}\label{coframe prep} The following two statements are immediate consequences of the definition of coprimality:

  \begin{enumerate}
\item For every element $p$ in a poset $L$, the characteristic map $$\chi_p:
  L\to{\sf 2}=\{0<1\}$$ of the up-set $\mathord\uparrow p$, defined by ($\chi_p(x)=1\iff p\leq x$) for all $x\in L$, preserves all (existing) infima. The map $\chi_p$ preserves finite suprema if, and only if, $p$ is coprime.

\item Let $X$ be any subset of the poset $L$. Then every element in $L$ is a supremum of elements in $X$ if, and only if, the following condition holds for all $x,y\in L$:
$$\forall p\in X\,(p\leq x \Rightarrow p\leq y)\Longrightarrow x\leq y;$$
equivalently, $x\nleq y$ only if for some $p\in X$ one has $p\leq x,$ but $p\nleq y.$
\end{enumerate}
\end{rem}

\begin{prop}\label{coframe}
If a complete lattice $L$ is sup-generated by its coprime elements, then it is a coframe, that is: finite suprema distribute over arbitrary infima in $L$.
\end{prop}

\begin{proof}(See  Theorem I-3.15 of \cite{Gierz2003}; the fact that in \cite{Gierz2003} the bottom element is considered coprime has no bearing on the validity of the statement.) By Remark \ref{coframe prep}(1), the map
$$\chi: L\to\prod_{p\in L\;\text{coprime}}{\sf 2},\quad x\mapsto (\chi_p(x))_p,$$
preserves arbitrary infima and finite suprema. Furthermore, by Remark \ref{coframe prep}(2), it is an injective map. Consequently, $L$ is isomorphic to a subcoframe of a power of the coframe {\sf 2}.
\end{proof}

Rewriting the codomain of the map $\chi$ as the powerset of $X=\{p\in L\,|\,p \;\text{coprime}\}$, under the hypothesis of the Proposition we can re-interpret $L$ as the lattice of closed sets of a topology on $X$. In the language of (co)locale theory (see \cite{Johnstone1982}), this means precisely that $L$ is a {\em spatial coframe}. Explicitly then, let us re-state (the dual of) Exercise 1.5 in \cite{Johnstone1982}, as follows:

\begin{cor}\label{spatial}
A complete lattice is sup-generated by its coprime elements if, and only if, it is a spatial coframe.
\end{cor}

\begin{proof}
For the ``if" part, let us just note that a spatial coframe $L$ can be thought of as the set of closed sets of a topological space X. To see then that every $A\in L$ is the join of coprime elements, so that $$A=\overline{\bigcup\{P\in L\;|\;P\; \text{coprime}, P\subseteq A\}},$$ it suffices to note that for every $x\in A$, the set $\overline{\{x\}}$ is coprime in $L$.
\end{proof}

Since complete lattices that are sup-generated by their coprime elements have the distributivity property of a coframe, it is natural to ask when such lattices may be {\em completely distributive}; more precisely, since so far we were able to avoid any use of the Axiom of Choice, we would like to know when they are {\em constructively completely distributive (ccd)} (see \cite{Wood2004}).
 Recall that a complete lattice $L$ is ccd if every element $a\in L$ is the join of all elements $x\ll a$ (``$x$ totally below $a$"); here $x\ll a$ means that, whenever $a\leq \bv B$ for $B\subseteq L$, then $x\leq b$ for some $b\in B$. For ccd complete lattices to be completely distributive (cd) in the classical sense, one needs the Axiom of Choice (AC); in fact the validity of (AC) is equivalent to ((ccd)$\Leftrightarrow$ (cd)) holding for all complete lattices: see \cite{Wood2004}.

 To answer the question raised, recall that $L$ (which, in general, may just be a poset) is {\em continuous} if every element $a\in L$ is the directed join of all elements $x\prec\!\!\prec a$ (``$x$ way below $a$"); here $x\prec\!\!\prec a$ means that, whenever $a\leq \bv D$ with $D\subseteq L$ directed, then $x\leq d$ for some $d\in D$. Without reference to (AC) one may still state the following Proposition:

 \begin{prop}\label{ccd}
 If the complete lattice $L$ is continuous and sup-generated by its coprime elements, then $L$ is constructively completely distributive.
 \end{prop}

 \begin{proof}
 Every $a\in L$ is the (directed) join of all $x\prec\!\!\prec a$, with each $x$ being the join of all coprime elements $p\leq x$; hence, $a=\bv\{p\in L\,|\, p\;\text{coprime}, \exists x\,(p\leq x\prec\!\!\prec a)\}$.
 It suffices to note now that each such $p$ is totally below $a$. Indeed, if $a\leq \bv B$, since $\bv B=\bv\{\bv F\,|\,F\subseteq B \;\text{finite}\}$ is a directed join, one first has $x\leq\bv F$ for some finite $F\subseteq B$, and then $p\leq b$ for some $b\in F$, by coprimality of $p$.
 \end{proof}

 It is well known that, with (AC) now granted, the sufficient condition for (ccd) of Proposition \ref{ccd} is also necessary:

 \begin{thm}
 Under the Axiom of Choice, a complete lattice is completely distributive if, and only if, it is a continuous spatial coframe.
 \end{thm}

 \begin{proof}
 For the part of the proof not yet covered by Proposition \ref{ccd} and Corollary \ref{spatial}, we refer to Theorem I-3.16 in \cite{Gierz2003}.
 \end{proof}

 We do not know whether there is a ``constructive version" of this theorem, that is: whether one can prove without invoking AC the converse statement of Proposition \ref{ccd}, so that a complete ccd lattice is a continuous spatial coframe.

\section{$\sV\text{-}{\bf Top}$ as a topological category}

From the presentation \ref{VTopasPVCat}(3) we know that the forgetful functor $|\text{-}|: \sV\text{-}{\bf Cls}\to {\bf Set}$ is topological (see \cite{MonTop}), a fact that may easily be checked also directly, as follows.

\begin{lem} For a family (of any size) of maps $f_i: X\to Y_i$ from a set $X$ into $\sV$-valued closure spaces $(Y_i,d_i),\,(i\in I)$, the $|\text{-}|$-initial structure $c$ on $X$ is given by
$$ (cA)(x)=\bw_{i\in I}d_i(f_iA)(f_ix)$$
for all $x\in X,\, A\subseteq X$.
\end{lem}

\begin{proof} (R) holds trivially, and for (T) we note that, for all $x\in X$ and $A,B\subseteq X$,
\begin{align*}
(\bw_{y\in B}(cA)(y))\otimes (cB)(x) &=(\bw_{y\in B}\bw_{i\in I}d_i(f_iA)(f_iy))\otimes\bw_{i\in I}d_i(f_iB)(f_ix) \\
&\leq\bw_{i\in I}((\bw_{y\in B}d_i(f_iA)(f_iy))\otimes d_i(f_iB)(f_ix))  \\
&\leq\bw_{i\in I}d_i(f_iA)(f_ix)=(cA)(x).
\tag*{\qEd}
\end{align*}
\def\popQED{}
\end{proof}
In order for us to conclude that the full subcategory $\sV\text{-}{\bf Top}$ of $\sV\text{-}{\bf Cls}$ is topological over $\bf Set$ as well, it suffices to show that it is {\em bicoreflective} (=coreflective, with all coreflections being {\em bimorphisms, i.e.}, both epic and monic) in $\sV\text{-}{\bf Cls}$. To this end,
for a subset $A$ of $X$, let ${\rm FinCov}(A)$ denote the set of finite covers of $A$, {\em i.e.}, of strings $(M_1,...,M_m)$ with $M_1\cup...\cup M_m=A$; here $m=0$ (the empty string $\emptyset$) is permitted when $A=\emptyset$. With the usual ``finer" relation
\[(M_1,...,M_m)\leq(N_1,...,N_n)\iff\forall i\in\{1,...,m\}\exists j\in\{1,...,n\}\;(M_i\subseteq N_j),\]
${\rm FinCov}(A)$ becomes a down-directed preordered set, {\em i.e.}, a preordered set in which finite sets have lower bounds.
For a $\sV\text{-}$valued closure space structure $c$ on $X$ and $\overrightarrow{M}=(M_1,...,M_m)\in {\rm FinCov}(A),\;x\in X$, let us write $$(c\overrightarrow{M})(x):=(cM_1)(x)\vee...\vee(cM_m)(x)$$ and define the {\em finitely additive core} $c^+$ of $c$ by
$$(c^+A)(x)=\bw_{\overrightarrow{M}\in{\rm FinCov}(A)}(c\overrightarrow{M})(x).$$

\begin{thm}\label{theorem}
Let the quantale $\sV$ be sup-generated by its coprime elements.
Then $(X,c)\mapsto (X,c^+)$ describes a right adjoint functor of $\sV\text{-}{\bf Top}\hookrightarrow\sV\text{-}{\bf Cls}$ which commutes with the underlying {\bf Set}-functors and has its counits mapping identically.
\end{thm}

\begin{proof}
Trivially, for $x\in X,\,A\subseteq X$, since $(A)\in{\rm FinCov}(A)$, one has $(c^+A)(x)\leq (cA)(x)$. Also, for all $\overrightarrow{M}=(M_1,...,M_m)\in {\rm FinCov}(A)$, $x\in A$ implies $x\in M_i$ for some $i$ and, hence, ${\sf k}\leq (cM_i)(x)\leq (c\overrightarrow{M})(x)$, which shows that $c^+$ is reflexive: ${\sf k}\leq (c^+A)(x)$.

To verify that $c^+$ is finitely additive, first note that, since the empty string covers the empty set, we have $\bot=(c\overrightarrow{\emptyset})(x)\geq (c^+\emptyset)(x)$ for all $x\in X$. Furthermore, for $A,B\subseteq X$ and $(M_1,..,M_m)\in{\rm FinCov}(A),
(N_1,...,N_n)\in{\rm FinCov(B)}$, one has $(M_1,...,M_m,N_1,...,N_n)\in{\rm FinCov}(A\cup B)$. Consequently, since $\sV$ is a coframe by Proposition \ref{coframe}, with the repeated application of the corresponding distributivity property (but exploited only for down-directed infima), one obtains
\begin{align*}
c^+(A\cup B)(x)&\leq\bw_{\overrightarrow{M}}\bw_{\overrightarrow{N}}((c\overrightarrow{M})(x)\vee(c\overrightarrow{N})(x))\\
&=\bw_{\overrightarrow{M}}((c\overrightarrow{M})(x)\vee(\bw_{\overrightarrow{N}}(c\overrightarrow{N})(x)))\\
&=\big(\bw_{\overrightarrow{M}}(c\overrightarrow{M})(x)\big)\vee(c^+B)(x)\\
&=(c^+A)(x)\vee(c^+B)(x).
\end{align*}
The reverse inequality holds since $c^+$ obviously inherits the monotonicity from $c$ (even when $\sV$ is not integral).

The crucial argument of the proof, namely the verification of (T) for $c^+$, can be given very compactly, as follows. From $c^+\leq c$ we conclude $d:=\overline{c^+}\leq \overline{c}\leq c$ with Lemma \ref{cc}.
Since, by Proposition \ref{addpres}, $d$ is finitely additive, $d\leq c$ trivially implies $d\leq c^+$ (as outlined more generally just below). This, by Lemma \ref{cc} again, means that $c^+$ is transitive.

Finally, to verify the adjunction, we must show that, for $(Y,d)\in\sV\text{-}{\bf Top}$, every continuous map $g:(Y,d)\to(X,c)$ is also continuous as a map $(Y,d)\to(X,c^+)$. But for $C\subseteq Y$ and all $(M_1,...,M_m)\in{\rm FinCov}(gC)$ one trivially has $(g^{-1}M_1\cap C,...,g^{-1}M_m\cap C)\in{\rm FinCov}(C)$ and therefore, by the finite additivity of $d$, for every $y\in Y$,
\begin{align*}
(dC)(y)&\leq d(g^{-1}M_1)(y)\vee...\vee d(g^{-1}M_m)(y)\\
&\leq c(gg^{-1}M_1)(gy)\vee...\vee c(gg^{-1}M_m)(gy)\\
&=(cM_1)(gy)\vee...\vee (cM_m)(gy).
\end{align*}
Consequently, $(dC)(y)\leq c^+(gC)(gy)$ follows, as desired.
\end{proof}

\begin{cor}
Let the quantale $\sV$ be a spatial coframe. Then the forgetful functor
$\sV\text{-}{\bf Top}\to{\bf Set}$ is topological, with initial liftings to be formed by coreflecting the initial lifting with respect to $\sV\text{-}{\bf Cls}\to{\bf Set}$ into $\sV\text{-}{\bf Top}$. Accordingly, limits in $\sV\text{-}{\bf Top}$ are formed by applying the coreflector to the corresponding limits formed in
$\sV\text{-}{\bf Cls}.$
 \end{cor}

 \begin{rem}\label{discrete}
 As topological functors, the underlying $\bf Set$-functors of $\sV\text{-}{\bf Cls}$ and $\sV\text{-}{\bf Top}$ have both, a full and faithful left adjoint and a full and faithful right adjoint, given by the discrete and indiscrete structures, respectively. But these are available without any extra provisions on $\sV$. The discrete $\sV$-valued closure structure on a set $X$ is given by the map
 \[c_{\rm{disc}}: \sP X\to \sV^X,\; (c_{\rm{disc}} A)(x)={\sf k} \;\;\text{if}\; x\in A, \text{ and }\; \bot\; \text{otherwise},  \]
 which is already finitely additive. The indiscrete $\sV$-valued closure space structure $c_{\rm{indisc}}$ maps $A\subseteq X$ to the constant function with value $\top$, but has to be corrected in case $A=\emptyset$ to the constant function $\bot$ in order to give the indiscrete $\sV$-valued topological structure on $X$.
  \end{rem}

\section{A $\sV$-categorical presentation of $\sV$-valued closure and topological spaces}

Recall that a (small) $\sV$-{\em category} is given by a set $X$ of objects and a ``hom map" $\text{hom}_X=a:X\times X\to\sV$ satisfying the conditions ${\sf k}\leq a(x,x)$ and $a(y,z)\otimes a(x,y)\leq a(x,z)$ for all $x,y,z\in X$.
A $\sV$-{\em functor} $f:(X,a)\to (Y,b)$ is a map $f:X\to Y$ with $a(x,y)\leq b(fx,fy)$ for all $x,y\in X$. The resulting category is $\sV\text{-}{\bf Cat}$, which is topological over $\bf Set$. There are three particular $\sV$-categories that will be used in what follows:

\begin{enumerate}
\item $\sV$ itself becomes a $\sV$-category with its hom-map $[-,-]$, characterized by
\[u\leq[v,w] \iff u\otimes v\leq w\]
for all $u,v,w\in\sV$, so that every $[v,-]$ is right adjoint to $(-)\otimes v:\sV\to\sV$ (as monotone maps).

\item For every set $X$, $\sV^X$ is the $X$-th power of $\sV$ in $\sV\text{-}{\bf Cat}$ when provided with the $\sV$-category structure
\[ [\sigma,\tau]=\bw_{x\in X}[\sigma x, \tau x]\;\;\;(\sigma ,\tau\in\sV^X).\]
We note that, of course, $X\mapsto \sV^X$ gives a functor ${\bf Set}^{\rm op}\to\sV\text{-}{\bf Cat}$ when one assigns to a map $f:X\to Y$ the $\sV$-functor \[f^!:\sV^Y\to\sV^X,\;\sigma\mapsto \sigma f.\]

\item For every set $X$, the power set $\sP X$ becomes a $\sV$-category when provided with the initial structure induced by the map $c_{\rm{disc}}:\sP X\to\sV^X$ of Remark \ref{discrete} (with respect to the underlying $\bf Set$-functor of $\sV\text{-}{\bf Cat}$), so that
\[\text{hom}_{\sP X}(A,B)=[c_{\rm{disc}} A,c_{\rm{disc}} B]\;\;\;(A,B\subseteq X).\]
\end{enumerate}

\begin{prop}\label{equivalent}
The following conditions are equivalent for a map $c:\sP X\to\sV^X$:
\begin{enumerate}[leftmargin=8mm]
\item[{\rm(i)}] $c$ is $\sV$-closure space structure on $X$;
\item[{\rm(ii)}] $c_{\rm{disc}} A\leq cA$ and $[c_{\rm{disc}} B,cA]\leq [cB,cA]$, \;for all $A,B\subseteq X$;
\item[{\rm(iii)}] $[c_{\rm{disc}} B, cA]=[cB, cA]$,\; for all $A,B\subseteq X$.
\end{enumerate}
\end{prop}

\begin{proof} \leavevmode
  \begin{description}[before={\renewcommand\makelabel[1]{##1:}}]
\item[(i)$\Leftrightarrow$(ii)] By definition of $c_{\rm{disc}} A$, (R) is equivalent to $(c_{\rm{disc}} A)(x)\leq (cA)(x)$ for all  $x\in X$. (T) may be written equivalently as
\[ \bw_{y\in B}(cA)(y)\leq\bw_{x\in X}[(cB)(x),(cA)(x)].\]
Since $[\bot,w]=\top$ and $[{\sf k},w]=w$ for all $w\in \sV$, the left-hand side of this inequality equals $\bw_{x\in X}[(c_{\rm{disc}} B)(x),(cA)(x)]$, so that (T) then becomes equivalent to the second inequality of (ii).

\item[(ii)$\Rightarrow$(iii)] $c_{\rm{disc}} B\leq cB$ gives $[c_{\rm{disc}} B, cA]\geq[cB, cA]$, since $[-,w]$ reverses the order for all $w\in \sV$.

\item[(iii)$\Rightarrow$(ii)] Since $(\sk\leq[v,w]\iff v\leq w)$ for all $v, w\in \sV$, from $[c_{\rm{disc}} A, cA]=[cA, cA]\geq{\sf k}$ one obtains $c_{\rm{disc}} A\leq cA$.
  \qedhere
\end{description}
\end{proof}

\begin{cor}\label{Vfunctor}
Every $\sV$-valued closure space structure $c$ on a set $X$ gives a $\sV$-functor $c:\sP X\to\sV^X$.
\end{cor}
\begin{proof}
Since $[v,-]$ is monotone for all $v\in\sV$, Proposition \ref{equivalent} gives $[c_{\rm{disc}} B, c_{\rm{disc}} A]\leq[c_{\rm{disc}} B,cA]\leq [cB,cA]$, for all $A,B\subseteq X$.
\end{proof}

For every $\sV$-category $X=(X,a)$, one has the {\em Yoneda $\sV$-functor}
\[{\sf y}_X:X\to\sV^X,\; x\mapsto {\rm hom}_X(-,x)=a(-,x).\]
We can now state:
\begin{cor}\label{composite}
For a set $X$, a map $c:\sP X\to \sV^X$ is a $\sV$-valued closure space structure on $X$ if, and only if,
\vspace{-2mm}
\[( \sP X\to^c\sV^X\to^{{\sf y}_{\sV^X}}\sV^{\sV^X}\to^{c_{\rm{disc}}^!}\sV^{\sP X})=(\sP X\to^c\sV^X\to^{{\sf y}_{\sV^X}}\sV^{\sV^X}\to^{c^!}\sV^{\sP X})\]
in $\sV\text{-}{\bf Cat}$ or, equivalently, in $\bf Set$. The map $c$ makes $X$ a $\sV$-valued topological space if, and only if, it
preserves finite suprema.
\end{cor}
\begin{proof}
The equality of the two composite maps simply rephrases condition (iii) of Proposition \ref{equivalent}.
\end{proof}

\begin{rem}
Since, by the {\em Yoneda Lemma}, or by an easy direct inspection, $[c_{\rm{disc}}\{x\},\sigma]=\sigma(x)$ for all $\sigma\in\sV^X, x\in X$, so that $(\{-\}\cdot c_{\rm{disc}})^!\cdot{\sf y}_{\sV^X}={\rm id}_{\sV^X}$,  the map $c$ may be recovered from the composite map of Corollary \ref{composite}, as
\[c= ( \sP X\to^c\sV^X\to^{{\sf y}_{\sV^X}}\sV^{\sV^X}\to^{c_{\rm{disc}}^!}\sV^{\sP X}\to^{\{-\}^!}\sV^X),\]
with $\{-\}:X\to\sP X$.
\end{rem}

In conclusion of this section, we see that the syntax needed to define $\sV$-valued closure spaces and $\sV$-valued topological spaces can be seen as living in $\sV$-{\bf Cat}, and that the axioms defining them are equational. Hence, when we consider these objects together with {\em closure preserving} maps as their morphisms (so that the  inequality of the continuity condition (C) becomes an equality), we obtain categories that are equationally defined within the $\sV$-{\bf Cat} environment. In particular, topological spaces, with closed continuous maps as their morphisms, form a category that is equationally defined within the realm of {\bf Ord}, the category of preordered sets and monotone maps.
$$ $$

\end{document}